\documentclass[a4paper,12pt]{article}
\usepackage[utf8]{inputenc}
\usepackage{amsmath}  
\usepackage{amssymb}       
\usepackage{amsfonts}
\usepackage{dsfont}
\usepackage{times}
\usepackage{amsthm} 
\usepackage{epsfig}
\usepackage{paralist}         
\usepackage{color}
\usepackage{cancel}               
\newcommand{\DATUM}{31.05.2026}       
\newcommand{\ol}{\overline}

\newcommand{\eps}{{\varepsilon}}       
          
\newcommand{\Om}{\Omega}               

\newcommand{\la}{\langle}
\newcommand{\ra}{\rangle}

\newcommand{\dGamma}{{\mathrm{d}\Gamma}}
\newcommand{\Bog}{{\mathrm{Bog}}}

\newcommand{\bfone}{\mathbf{1}}
\newcommand{\cB}{\mathcal{B}}
\newcommand{\cD}{\mathcal{D}}
\newcommand{\cE}{\mathcal{E}}

\newcommand{\cI}{\mathcal{I}}

\newcommand{\cL}{\mathcal{L}}         
         
\newcommand{\cN}{\mathcal{N}}         
\newcommand{\cO}{\mathcal{O}}         
         
\newcommand{\cQ}{\mathcal{Q}}

\newcommand{\cS}{\mathcal{S}}
\newcommand{\cT}{\mathcal{T}}
\newcommand{\cU}{\mathcal{U}}

\newcommand{\bbA}{{\mathds{A}}} 
 
\newcommand{\bbC}{{\mathds{C}}}

\newcommand{\bbN}{{\mathds{N}}}     
\newcommand{\bbP}{{\mathds{P}}}   
 
\newcommand{\bbR}{{\mathds{R}}}     
\newcommand{\bbZ}{{\mathds{Z}}}     

\newcommand{\bbU}{{\mathds{U}}} 
   
\newcommand{\bbW}{{\mathds{W}}}     
\newcommand{\bbone}{{\mathds{1}}}
\newcommand{\fF}{\mathfrak{F}}

\newcommand{\fh}{\mathfrak{h}}

\newcommand{\sfJ}{{\mathsf J}}

\newcommand{\hA}{\widehat{A}}

\newcommand{\hcE}{\widehat{\mathcal{E}}}

\newcommand{\hxi}{\widehat{\xi}}
\newcommand{\tcE}{\widetilde{\cE}}

\newcommand{\tH}{\widetilde{H}}

\newcommand{\vG}{{\vec{G}}}

\newcommand{\veps}{{\vec{\varepsilon}}}

\newcommand{\vf}{{\vec{f}}}

\newcommand{\vk}{{\vec{k}}}

\newcommand{\vp}{{\vec{p}}}

\newcommand{\vx}{{\vec{x}}}
\newcommand{\vO}{{\vec{0}}}
\newcommand{\DM}{\mathfrak{DM}}

\newcommand{\rRe}{\mathrm{Re}}

\newcommand{\cirS}{\mathop{\bigcirc\kern -.73em {\scriptstyle{\rm S}}}}

\newcommand{\Tr}{{\rm Tr}}
 
\newcommand{\op}{\mathrm{op}} 
\newcommand{\gs}{{\rm gs}}

\newcommand{\el}{{\mathrm{el}}}
\newcommand{\ph}{{\mathrm{ph}}}
\newcommand{\BHF}{{\small{\mathrm{BHF}}}}

\newcommand{\rmmin}{{\mathrm{min}}}
\newcommand{\LL}{{\mathrm{LL}}}

\newcommand{\LBOUND}{E_\mathrm{low}}
\newcommand{\UBOUND}{E_\mathrm{up}}

\usepackage{enumitem}
\numberwithin{equation}{section}

\newtheorem{thm}{Theorem} [section]  
\newtheorem{lem}[thm]{Lemma}

\theoremstyle{plain}
\begin{document}
\bibliographystyle{plain}

\title{Bounds on the Bogoliubov--Hartree--Fock Energy of the Pauli--Fierz
  Hamiltonian}

\author{
Volker Bach \\
\small{IAA, TU Braunschweig, Germany (v.bach@tu-bs.de)}
\and
Matthias Herdzik \\
\small{IAA, TU Braunschweig, Germany (m.herdzik@tu-bs.de)}
}
\date{\DATUM}

\maketitle

\vspace*{-12mm}

\begin{center}
Dedicated to Israel Michael Sigal, \\
with admiration for his creativity.
\end{center}

\textbf{Abstract:} A variational analysis of the
Bogoliubov--Hartree--Fock (BHF) energy of the translation-invariant,
spinless Pauli--Fierz Hamiltonian with massless dispersion relation
built up on \cite{BachBreteauxTzaneteas2013} and \cite{BachHach2022}
is presented. The main results are lower and upper bounds on the BHF
energy for fixed total momentum expressed through simpler variational
problems defined on the space of positive Hilbert--Schmidt operators
and a new variational formulation of the upper bound for zero total
momentum. Specifically, we introduce a change of variables which
considerably simplifies the energy functional and the derivation of
its stationarity condition.

\section{Introduction} \label{sec-I}
%
In recent decades the ultraviolett problem of a single
non-relativistic, spinless particle coupled to the radiation field has
been analysed using different methods. Specifically relevant for the
present work are \cite{LiebLoss2005, BachBreteauxTzaneteas2013,
  BachHach2022}, from which we review some parts after introducing the
model.

\vspace{1em}

Let $\tilde{\fh} := 
\big\{ \vf \in L^2(S_{\sigma,\Lambda}; \bbC \otimes \bbR^3)
\, \big| \: \forall \vk \in S_{\sigma,\Lambda} \; a.e.: \ 
\vk\cdot \vf(\vk) = 0 \big\}$ 
be the Hilbert space of square-integrable, transverse vector fields
defined on the \textit{momentum range}
\begin{align} \label{eq-I.01}
S_{\sigma,\Lambda} \ := \ 
\big\{ \vk \in \bbR^3 \; : \ \sigma \leq |\vk| \leq \Lambda \big\} \, ,
\end{align} 
where $0 < \sigma < \Lambda < \infty$ are the infrared and
ultraviolett cutoffs. For $\vk \in S_{\sigma,\Lambda}$, the vectors
$\veps_+(\vk)$ and $\veps_-(\vk)$ are chosen in a way such that
$\{\veps_+(\vk),\veps_-(\vk), \vk/|\vk| \} \subseteq \bbR^3$ forms
a real orthonormal basis and such that $\vk \mapsto \veps_+(\vk)$ and
$\vk \mapsto \veps_-(\vk)$ are measurable. 
This allows us to identify $\tilde{\fh}$ with the 
\textit{one-photon Hilbert space} 
\begin{align} \label{eq-I.03}
\fh \ := \ L^2(S_{\sigma,\Lambda}\times \bbZ_2) 
\end{align}
by virtue of the unitary map 
\begin{align} \label{eq-I.04}
\fh \, \ni \, f(\vk,\tau) 
\ \mapsto \ 
\veps_+(\vk) \, f(\vk,+) \, + \, \veps_-(\vk) \, f(\vk,-) 
\, \in \, \tilde{\fh} \, .
\end{align}
Additionally, we assume 
that, for any $\vk \in S_{\sigma,\Lambda}$ 
\begin{align} \label{eq-I.02}
\veps_\pm (-\vk) \ = \ -\veps_\pm (\vk) \, .
\end{align}
The \textit{Pauli--Fierz Hamiltonian} is the selfadjoint operator 
\begin{align} \label{eq-I.05}
\tH_g \ := \ 
\frac{1}{2} \big( \tfrac{1}{i}\vec{\nabla} 
+ \vec{\bbA}(\vec{x})\big)^2 \: + \: H_\ph \, ,
\end{align}
which is defined on 
$H^1(\bbR^3) \otimes \cD(\cN^{1/2}) \subseteq 
L^2(\bbR^3) \otimes \fF_\ph$ as a quadratic form,
$\fF_\ph := \fF_b(\fh)$ being the boson Fock space over $\fh$, 
where the \textit{photon field energy} is defined as
\begin{align} \label{eq-I.06}
H_\ph \ := \ \dGamma(|k|)
\end{align}
and the \textit{magnetic vector potential} is given by
\begin{align} \label{eq-I.07}
\vec{\bbA}(\vx)
\ := \ 
a^*\big(e^{-i \vk \cdot \vx} \vG \big) \, + \, a\big(e^{-i \vk \cdot \vx} \vG \big) \, ,
\end{align}
where 
\begin{align} \label{eq-I.08}
\vG(\vk, \tau) \ := \  
\veps_\tau(\vk) \; \frac{g }{|\vk|^{1/2}} \, ,
\end{align} 
and $g \geq 0$ is the coupling constant. The \textit{ground state
  energy} is the infimum
\begin{align} \label{eq-I.09}
E_\gs \ := \ 
\inf\Big\{ \big\la \Psi \big| \tH_g \Psi \big\ra \: \Big|
\ \Psi \in H^1(\bbR^3) \otimes \cD(\cN^{1/2}), \ \|\Psi\|=1 \Big\} \, ,
\end{align}
which, by the Rayleigh--Ritz principle is equal to the bottom of the
spectrum $\inf\sigma(\tH_g)$. 

In \cite{LiebLoss2005}, Lieb and Loss derived the following lower and
upper bounds
\begin{align} \label{eq-I.10}
C_1 \, \alpha^{2/7} \, \Lambda^{3/2} \ \leq \ 
E_\gs \ \leq \ C_2 \, \alpha^{2/7} \, \Lambda^{12/7} \, ,
\end{align}
on the ultraviolett behavior, $\Lambda \gg 1$, of the ground state
energy. Here, $0 < C_1, C_2 <\infty$ are constants and 
$\alpha := g^2 \geq 0$ is the fine structure constant. They
furthermore conjectured the ultraviolett behaviour 
$E_\gs \sim \alpha^{2/7} \Lambda^{12/7}$, as $\Lambda \to \infty$. The
significance of this result lies in the discrepancy to the prediction
$E_\gs \sim \alpha \Lambda^2$ one obtains from naive perturbation
theory about the vacuum vector, thus stressing the importance of
non-perturbative methods. The upper bound in \eqref{eq-I.10} was
obtained through the analysis of the related variational problem
\begin{align} \label{eq-I.11}
& E_\LL \ := \ 
\\ \nonumber 
& \inf\Big\{ \la \Psi | \tH_g \Psi \ra  \: \Big| \ 
\Psi = \varphi_\el \otimes \psi_\ph \in H^2(\bbR^3) \otimes \cD(\cN), 
\ \|\varphi_\el\| = \|\psi_\ph\| = 1 \Big\} \, ,
\end{align}
which we call the \textit{Lieb--Loss Energy}. Note that 
the Lieb--Loss energy is an upper bound on the ground state
energy, although it is a priori unclear how large the deviation is. 

In \cite{BachHach2022}, Bach and Hach extended the methods from
\cite{LiebLoss2005} and proved the conjecture of Lieb and Loss with a
quantitative error bound,
\begin{align} \label{eq-I.12}
-C \, \alpha^{4/49} \, \Lambda^{-4/49} \ \leq \ 
\frac{E_{\LL}}{F \, \alpha^{2/7} \, \Lambda^{12/7}} \, - \, 1 
\ \leq \ 
C \, \alpha^{4/105} \, \Lambda^{-4/105},
\end{align}
where $C>0$ is a universal constant and $F > 0$ is the zero
of a Bessel function.

A different route was taken by Bach, Breteaux, and Tzaneteas in
\cite{BachBreteauxTzaneteas2013}. To begin with, the
translation-invariance allows to remove the particle degree
of freedom by conjugating the Pauli--Fierz Hamiltonian by
a suitable unitary transformation $\bbU$ which yields the direct
integral decomposition 
$\bbU \tH_g \bbU^*=\int^\oplus H_{g,\vec{p}} \, \mathrm{d}^3p$, where
\begin{align} \label{eq-I.13}
H_{g,\vp} \ := \ 
\frac{1}{2}\big( \vec{\bbP}_\ph + \vec{\bbA}(\vO) - \vp \big)^2 
\: + \: H_\ph
\end{align}
is the \textit{fiber Hamiltonian of total momentum $\vp$}, and 
\begin{align} \label{eq-I.14}
\vec{\bbP}_\ph \ := \ \dGamma(\vk)
\end{align}
is the \textit{momentum of the photon field}. The operators
in both \eqref{eq-I.13} and \eqref{eq-I.14} are defined on 
$\cD(\cN)\subseteq \fF_\ph$. Observing that
\begin{align} \label{eq-I.17}
E_\gs
\ = \ 
\inf_{\vp \in \bbR^3} E_\gs(\vp) \, ,
\end{align}
where
\begin{align} \label{eq-I.15}
E_\gs(\vp) 
\ := \ 
\inf \sigma (H_{g,\vp}) 
\ = \ 
\inf\big\{ \Tr[\rho H_{g,\vp}] \; \big| 
\ \rho \in \DM \big\} \, ,
\end{align}
with
\begin{align}  \label{eq-I.16}
\DM
\ := \ 
\Big\{ \rho \in \cL^1(\fF_\ph) \: \Big| 
\ \rho\geq 0, \ \Tr[\rho]=1, \ \rho H_{g,\vp}, H_{g,\vp}\: \rho 
\in \cL^1(\fF_\ph) \Big\}
\end{align}
is the convex set of density matrices of finite energy, we are lead to
the \textit{Bogoliubov--Hartree--Fock (BHF) Approximation} of
$E_\gs(\vp)$ given for fixed $\vp \in \bbR^3$ by the \textit{BHF
  energy}
\begin{align} \label{eq-I.18}
E_\BHF(\vp) \ := \
\inf\big\{ \Tr[\rho H_{g,\vp}] \; \big| 
\ \rho \in \DM, \ \rho \text{ is quasifree} \big\} \, .
\end{align}
Quasifree states are those states, which are fully characterized by
their two-point functions $\Tr[\rho\: a^*(f)a(g)]$ and for
Hamiltonians, which are quadratic in the fields, the BHF approximation
is exact, see \cite{BachLiebSolovej1994}. The Pauli--Fierz Hamiltonian,
however, has quartic parts, therefore, $E_\BHF(\vp)$ is at least a
priori a true approximation of $E_\gs(\vp)$. For the exact definition
of quasifree density matrices, we refer to the second section of
\cite{BachBreteauxTzaneteas2013}. The authors also showed
that
\begin{align} \label{eq-I.19}   
E_\BHF(\vp) \ = \ 
\inf\Big\{ \Tr[\rho H_{g,\vp}] \: \Big| \ 
\rho \in \DM, \ \rho \text{ is quasifree and pure} \Big\} \, ,
\end{align} 
and that all pure quasifree density matrices can be expressed through
Bogoliubov transformations, see Sect.~\ref{sec-II}, which is the
starting point of our analysis and which has also been shown in
\cite{DerezinskiNapiorkowskiSolovej2013}.

The ultimate goal is to determine the optimal
Bogoliubov transformation and to conjugate the fiber Hamiltonian by
it, which we hope will unveil the nature of the ultraviolett
singularity of the fiber and full Hamiltonian.

Our new results are variational upper and lower bounds on the BHF
energy, which we consider an important step towards the notoriously
difficult analysis of the ultraviolet limit in the instance of the
Pauli--Fierz model.
%
\begin{thm} \label{thm-I.1}
Let $\sfJ: \fh\to \fh$ be the antiunitary map $\psi(k) \mapsto
\overline{\psi(-k)}$. Then for any $\vp \in \bbR^3$
\begin{align} \label{eq-I.20}   
\LBOUND(\vp) \ \leq \ E_{\BHF}(\vp) \ \leq \ \UBOUND(\vp)
\end{align} 
where
\begin{align} \label{eq-I.21}   
\LBOUND(\vp) \ = \ &
\inf\Big\{ \cE_{g,\vp}(V, \eta) \; \Big| \ 
V \in \cL_{s.a.}^2(\fh), \ V \geq 0, \ \eta \in \fh \Big\} \, ,
\\[1ex] \label{eq-I.22}   
\UBOUND(\vp) \ = \ &
\inf\Big\{ \cE_{g,\vp}(V, \eta) \; \Big| \ 
V = \sfJ V\sfJ  \in \cL_{s.a.}^2(\fh), \ V \geq 0, \ 
\eta = \sfJ \eta \in \fh \Big\} \, ,
\end{align} 
with
\begin{align} \label{eq-I.23}   
\cE_{g,\vp}(V,\eta) 
\ = \ 
\frac{1}{2} \sum\nolimits_{\nu=1}^3 & \Big\{ 
\Big(\Tr[ k_\nu V^2 ] + 
\la \eta | k_\nu \eta \ra +2 \mathrm{Re}\la \eta | G_\nu \ra -p_\nu \Big)^2 
\nonumber \\ & 
- \Tr[(k_\nu V\sqrt{1+V^2} )^2] + \Tr [k_\nu V^2 k_\nu (1+V^2)]
\nonumber \\ & 
+ \la G_\nu +k_\nu \eta | (\sqrt{1+V^2}-V)^2 (G_\nu +k_\nu \eta) \ra \Big\} 
\nonumber \\ & 
+\Tr[|k|V^2] + \big\la \eta \big| |k| \eta \big\ra 
\end{align} 
and 
$\cL_{s.a.}^2(\fh) := \{ A \in \cL^2(\fh) | A= A^* \} \subseteq \cL^2(\fh)$
being the real subspace of self-adjoint Hilbert--Schmidt operators.
\end{thm}
%
Theorem~\ref{thm-I.1} implies that $\UBOUND(\vp)$ is an upper bound on
$E_\gs(\vp) \leq E_\BHF(\vp) \leq \UBOUND(\vp)$ and, therefore,
  also on $E_\gs$. 
We then focus on the analysis of this upper bound
$\UBOUND(\vO)$ for zero total momentum. In this case, we
reparametrize the variable $V$ and eventually eliminate the
variable $\eta$ from the variation by completion of a square in
$\cE_{g,\vO}$ as
\begin{align} \label{eq-I.24,1}   
\cE(z) \ := \ & \cE_{g,\vO}\big( V_z , \eta_z \big) \, ,
\quad \text{with} \quad
V_z \ := \ \frac{z}{2\sqrt{1+z}} 
\quad \text{and} 
\\[1ex] \label{eq-I.24,2}   
\eta_z \ := \ & 
- \frac{1}{2} \bigg( |k| + \frac{1}{2} \sum_{\nu=1}^3 
k_\nu (1+z)^{-1} k_\nu \bigg)^{-1} 
\sum_{\nu=1}^3 k_\nu (1+z)^{-1} G_\nu \, .
\end{align} 
The stationarity condition for the resulting functional $\cE$ 
is the content of the following theorem.
%
\begin{thm} \label{thm-I.2}
We have that 
\begin{align} \label{eq-I.25}   
\UBOUND(\vO) \ = \ 
\inf\Big\{ \cE(z) \; \Big| 
\ z = \sfJ z\sfJ  \in \cL_{\geq 0}^2(\fh) \Big\} \, .
\end{align} 
Furthermore, if $z = \sfJ z\sfJ \in \cL_{\geq 0}^2(\fh)$ is a
minimizer of $\cE$, i.e., if $\cE(z) = \UBOUND(\vO)$, then
\begin{align} \label{eq-I.26}   
2 \sum_{\nu=1}^3 \Big\{ 
|G_\nu + k_\nu \eta_z \ra \la G_\nu + k_\nu \eta_z | \Big\} 
\ = \  &
(1+z) A_z (1+z) \, - \,  A_0 \,-\frac{1}{2}\sum_{\nu=1}^3 k_\nu z k_\nu ,
\\ \label{eq-I.27}   
\text{where $\eta_z$ is as in \eqref{eq-I.24,2} and } \ \ A_z \ := \ & 
|k| \, + \, \frac{1}{2}\sum_{\nu=1}^3 k_\nu (1+z)^{-1} k_\nu \, .
\end{align} 
\end{thm}
%
Note that Eq.~\eqref{eq-I.26} is a statement about Hilbert--Schmidt
operators, even though $A_z$ is not of Hilbert--Schmidt class, but the
difference $(1+z)A_z(1+z)-A_0$ on the right side of \eqref{eq-I.26}
is. Further note that only the left side of \eqref{eq-I.26} depends on
the coupling constant $g$ and is an operator of rank three, hence
Hilbert--Schmidt.

In Sect.~\ref{sec-II}, we introduce Bogoliubov transformations and
derive an energy functional from the characterization~\eqref{eq-I.19}.
Sect.~\ref{sec-III} contains our variational analysis of the BHF
energy, employing the methods of \cite{BachHach2022} to derive the
aforementioned bounds in Theorem~\ref{thm-I.1}. Sect.~\ref{sec-IV}, is
devoted to the study of the parameterization 
$V_z = \frac{z}{2\sqrt{1+z}}$ and its properties, and in
Sect.~\ref{sec-V} we derive the Euler--Lagrange equations for the
simplified functional. 
%

\section{Description of the Problem} \label{sec-II} 
%
This section contains the definition of Bogoliubov transformations in
terms of Bogoliubov maps and their basic properties, as well as a
derivation of a functional for the energy of pure quasifree states.
%

\subsection{Antiunitary Involutions and CCR} \label{subsec-II.1} 
%
Recall from \eqref{eq-I.03} that $\fh = L^2(S_{\sigma, \Lambda} \times
\bbZ_2)$ is the one-photon Hilbert space and that the photon Fock
space $\fF_\ph$ is the boson Fock space $\fF_\ph = \fF_b(\fh)$ over
$\fh$, with number operator $\cN_\ph$.  We assume to be given the
family $\{a^*(f), a(f)\}_{f \in \fh}$ of creation and annihilation
operators which fulfill the canonical commutation relations (CCR)
\begin{align} \label{eq-II.00,1}   
\forall \, f, g \in \fh: \ \ 
[a^*(f), a^*(g)] \, = \, [a(f), a(g)] \, = \, 0 \, , \ \ 
[a(f), a^*(g)] \, = \, \la f | g \ra \, \bfone_\fF \, .
\end{align} 
Also, we require $a(f) \Om = 0$, for all $f \in \fh$.
As usual, $\cB(\fh)$ and $\cL^2(\fh)$ denote the space of bounded
and Hilbert--Schmidt operators, respectively. Given a real number
$\kappa \in \bbR$, we define the open and closed, resp., convex subsets 
\begin{align} \label{eq-II.01,1}   
\cL_{> \kappa}^2(\fh) \ := \ 
\big\{ A \in \cL_{s.a.}^2(\fh) \; \big| \ A > \kappa \big\} 
\ \subseteq \ \cL_{s.a.}^2(\fh) \, ,
\\[1ex] \label{eq-II.01,2}   
\cL_{\geq \kappa}^2(\fh) \ := \ 
\big\{ A \in \cL_{s.a.}^2(\fh) \; \big| \ A \geq \kappa \big\} 
\ \subseteq \ \cL_{s.a.}^2(\fh) \, ,
\end{align} 
of $\cL_{s.a.}^2(\fh)$ containing all self-adjoint
Hilbert--Schmidt operators bounded below by $\kappa$. 
For $k = (\vk,\tau) \in S_{\sigma,\Lambda} \times \bbZ_2$,
with $\vk = (k_1, k_2, k_3)$, we use the customary notation  
$-k := (-\vk, \tau)$ and $|k| := |\vk| = \sqrt{k_1^2 + k_2^2 + k_3^2}$. 
As in \cite{BachHach2022}, we define an antiunitary involution
$\sfJ: \fh \to \fh$ by
\begin{align} \label{eq-II.02}   
[\sfJ \eta](k) \ := \ \ol{\eta(-k)} \, ,
\end{align} 
i.e., a bijection on $\fh$ which fulfills $\sfJ^2 = \bfone_\fh$ and
\begin{align} \label{eq-II.03}
\sfJ( \beta \eta + \xi) \ = \ \ol{\beta} \sfJ(\eta) + \sfJ(\xi) \, , 
\quad 
\la \sfJ(\eta) | \sfJ(\xi) \ra \ = \ \la \xi | \eta \ra \, ,
\end{align} 
for all $\eta, \xi \in \fh$ and $\beta \in \bbC$. Note the special
property 
\begin{align} \label{eq-II.04}
\sfJ \, \vk \, \sfJ \ = \ - \vk
\end{align} 
of the antiunitary involution in \eqref{eq-II.02}. Equivalently,
$\sfJ[ \vk \eta](\vp, \tau) = - \vp \, \sfJ[\eta](\vp, \tau) 
= - \vp \, \ol{\eta(-\vp, \tau)}$,
for $\eta \in \fh$ and 
$(\vp, \tau) \in S_{\sigma,\Lambda} \times \bbZ_2$.

\subsection{Bogoliubov Transformations} \label{subsec-II.2} 
%
We now introduce \textit{Bogoliubov transformations} which are special
unitary transformations $\bbU \in \cU[\fF_\ph]$ on 
photon Fock space described below.

\vspace*{-2ex}

\paragraph{Bogoliubov Maps}
To begin with we consider affine linear maps on creation
  and annihilation operators of the form 
\begin{align} \label{eq-II.05}
\forall f \in \fh: \quad 
b^*(f) \ = \ 
a^*(U f) + a(\sfJ V f) + \la \eta | f \ra \, .
\end{align} 
These maps are parametrized by a vector $\eta \in \fh$ and two
operators $U, V \in \cB(\fh)$. The requirement that the new creation
and annihilation operators $\{b^*(f), b(f)\}_{f \in \fh}$ also fulfill
the CCR~\eqref{eq-II.00,1} leads to \textit{Bogoliubov maps} 
which are defined as
\begin{align} \label{eq-II.06}   
& \Bog_\sfJ'[\fh] 
\ := \ 
\\ \nonumber &
\bigg\{ B \equiv B(U,V) = 
\begin{pmatrix} 
U & \sfJ V \sfJ \\ V & \sfJ U \sfJ
\end{pmatrix} \in \cB(\fh \oplus \fh) \;
\bigg| \ B^* \cS B = \cS  \, , \ \ B \cS B^* = \cS \bigg\} \, ,
\end{align} 
where 
\begin{align} \label{eq-II.07}   
\cS \ := \ 
\begin{pmatrix} 
\bfone_\fh & 0 \\ 0 & - \bfone_\fh 
\end{pmatrix} \, .
\end{align} 
That is, the family $\{b^*(f), b(f)\}_{f \in \fh}$ defined by
\eqref{eq-II.05} fulfills the CCR~\eqref{eq-II.00,1} if, and only if,
$B(U,V) \in \Bog_\sfJ'[\fh]$ is a Bogoliubov map and $\eta \in \fh$.
Note that $\Bog_\sfJ'[\fh] \subseteq \cB(\fh \oplus \fh)$ is a 
subgroup of the automorphisms on $\fh \oplus \fh$ and that
$B(U,V) \in \Bog_\sfJ'[\fh]$ is a Bogoliubov map if, and only if,

\begin{align} \label{eq-II.08} 
U^* U \ = \ \bbone + V^* V \; , & \quad
U^* \sfJ V \ = \ V^* \sfJ U \, ,
\\[1ex] \label{eq-II.09} 
U U^* \ = \ \bbone + \sfJ V V^* \sfJ \; , & \quad
\sfJ U V^* \ = \ V U^* \sfJ \, .
\end{align} 

\vspace*{-2ex}

\paragraph{Proper Bogoliubov Maps}
According to the \textit{Shale--Stinespring condition} the
transformation $a^*(f) \mapsto b^*(f)$ in Eq.~\eqref{eq-II.05} 
determined by $B \equiv B(U,V) \in \Bog_\sfJ'[\fh]$ and with
$\eta =0$ can be implemented as a conjugation 
$b^*(f) = \bbU_{B} a^*(f) \bbU_{B}^* = a^*(U f) + a(\sfJ V f)$ 
by a unitary operator
$\bbU_{B} \in \cU[\fF_\ph]$ if, and only if, $V \in \cL^2(\fh)$
is a Hilbert--Schmidt operator on $\fh$. This leads us to introduce
\textit{proper Bogoliubov maps}
\begin{align} \label{eq-II.10}   
& \Bog_\sfJ[\fh] 
\ := \ 
\\ \nonumber &
\Big\{ B \equiv B(U,V) \in \Bog_\sfJ'[\fh] \; \Big| \ \ 
V \in \cL^2(\fh) \Big\} \ \subseteq \ \Bog_\sfJ'[\fh] \, ,
\end{align} 
which form a subgroup of $\Bog_\sfJ'[\fh]$.

\vspace*{-2ex}

\paragraph{Weyl Transformations} Conversely, if 
$B = \big( \begin{smallmatrix} \bfone & 0 
\\ 0 & \bfone \end{smallmatrix} \big)$ and $\eta \in \fh \setminus \{0\}$
then it is well-known that 
$b^*(f) = a^*(f) + \la \eta | f \ra = \bbW_{\eta} a^*(f)  \bbW_{\eta}^*$,
where $\bbW_\eta := \exp[a^*(\eta) - a(\eta)] \in \cU[\fF_\ph]$ is
the unitary \textit{Weyl operator}.

\vspace*{-2ex}

\paragraph{Bogoliubov transformations} 
Composing $\bbW_\eta$ and $\bbU_{B}$, we arrive at
the group of \textit{Bogoliubov transformations} 
$\bbU_{B} \bbW_\eta \in \cU[\fF_\ph]$ determined by a proper
Bogoliubov map $B = B(U,V) \in \Bog_\sfJ[\fh]$ and a vector 
$\eta \in \fh$, with
\begin{align} \label{eq-II.11}
\forall f \in \fh: \quad 
\bbU_{B} \, \bbW_\eta \, a^*(f) \, \bbW_\eta^* \, \bbU_{B}^* \ = \ 
a^*(U f) + a(\sfJ V f) + \la \eta | f \ra \, .
\end{align} 
Note that, given an orthonormal basis $\{f_n\}_{n=1}^\infty \subseteq \fh$, 
due to $a(U f_n) \Om = 0$,  
\begin{align} \label{eq-II.12} 
\la \bbW_\eta^* \, \bbU_B^* & \Om \: | 
\: \cN \, \bbW_\eta^* \, \bbU_B^* \Om \ra _{\fF_\ph}
\ = \ 
\sum_{n=1}^\infty \| \bbU_B \, \bbW_\eta \, 
a(f_n) \, \bbW_\eta^* \, \bbU_B^* \Om \|^2_{\fF_\ph}
\nonumber \\[1ex] 
\ = \ &
\sum_{n=1}^\infty \big\| \big( a^*(\sfJ V f_n) +
\la f_n | \eta \ra \big) \Om \big\|^2_{\fF_\ph}
\ = \ 
\sum_{n=1}^\infty \big\{ \| V f_n \|^2_{\fh} 
+ |\la \eta | f_n \ra_{\fh}|^2 \big\}
\nonumber \\[1ex] 
\ = \ &
\| V \|_{\cL^2(\fh)}^2 + \|\eta\|^2_{\fh} \, ,
\end{align} 
which shows that the Shale--Stinespring condition ensures that the
particle number expectation value of the vacuum vector stays finite
under Bogoliubov transformations.

\vspace*{-2ex}

\paragraph{Pure Quasifree States as Bogoliubov transforms 
of the Vacuum Projection} Another important fact we use is that the
set of pure quasifree states is the orbit of the vacuum projection
$|\Om\ra\la\Om|$ under Bogoliubov transformations. 
That is, $\rho \in \DM$ is quasifree and pure if, and only if, there
exists $B \in \Bog_\sfJ[\fh]$ and $\eta \in \fh$ such that 
$\rho = |\bbW_\eta^* \bbU_{B}^* \Om \ra \la \bbW_\eta^* \bbU_{B}^* \Om|$.
Together with \eqref{eq-I.19}, this implies that
\begin{align} \label{eq-II.13}   
E_\BHF(\vp) \ = \ 
\inf\big\{ \tcE_{g,\vp}(U, V, \eta) \; \big| \ 
B(U,V) \in \Bog_\sfJ[\fh], \ \eta \in \fh \big\} \, ,
\end{align} 
where
\begin{align} \label{eq-II.14}   
\tcE_{g,\vp}(U, V, \eta) \ := \ 
\big\la \Om \, \big| \, \bbU_{B(U,V)} \bbW_\eta 
H_{g,\vp} \, \bbW_\eta^* \bbU_{B(U,V)}^* \Om \big\ra_{\fF_\ph} \, .
\end{align} 
%
\begin{lem} \label{lem-II.1}
Let $g \in \bbR$ and $\vp \in \bbR^3$. For 
$B(U,V) \in \Bog_\sfJ[\fh]$ and $\eta \in \fh$, the BHF energy functional
is given by
\begin{align} \label{eq-II.15}  
\tcE_{g,\vp}(U,V,\eta)
\ = \ 
\cS_{g,\vp}(V,\eta) + \cT(U,V) + \cQ_{g}(U,V,\eta) + \cI(V,\eta) \, ,
\end{align} 
where the square term $\cS_{g,\vp}$, the trace term $\cT$, the
quadratic form term $\cQ_{g}$, and the field term $\cI$ are given by
\begin{align} \label{eq-II.16} 
\cS_{g,\vp}(V,\eta) 
\ := \ & 
\frac{1}{2} \sum_{\nu=1}^3 \Big\{ 
\Tr[k_\nu V^* V] + \la \eta| k_\nu \eta \ra 
+ 2 \rRe \la \eta |G_\nu \ra - p_\nu \Big\}^2 \, ,
\\ \label{eq-II.17} 
\cT(U,V) 
\ := \ & 
\frac{1}{4} \sum_{\nu=1}^3 \Big\{ 
\Tr\big[ (V^* \sfJ  U k_\nu)^2 \big] +
\Tr\big[k_\nu V^* V k_\nu (1+V^* V)\big] \Big\} \, ,
\\ \label{eq-II.18} 
\cQ_{g}(U,V,\eta)
\ := \ & 
\frac{1}{2} \sum_{\nu=1}^3 \Big\{ 
\big\la G_\nu + k_\nu \eta \: \big| \: 
(1+2V^* V) (G_\nu + k_\nu \eta) \big\ra
\nonumber \\ 
& \qquad \quad + 2 \rRe\big\la G_\nu + k_\nu \eta \:\big| \: 
V^* \sfJ U (G_\nu+k_\nu \eta) \big\ra \Big\} \, ,
\\ \label{eq-II.19}  
\cI(V,\eta) 
\ := \ &
\Tr\big( |k| V^* V \big) + \la \eta | |k|\eta \ra \, . 
\end{align} 
\begin{proof}
We omit the proof because it is merely a computation.
\end{proof}
\end{lem}
%
Note that this result does not depend on the special choice of $\sfJ$
and holds true for any antiunitary involution.

\section{Variational Analysis} \label{sec-III}
%
In this section we prove Theorem~\ref{thm-I.1}. Our proof borrows
ideas from \cite[Lemma~IV.8]{BachHach2022}. We begin with the 
following Lemma.
%
\begin{lem} \label{lem-III.1}
Let $B(U,V) \in \Bog[\fh]$ and $\eta \in \fh$. Then
\begin{align} \label{eq-III.01}   
\Tr\big\{ \big( k_\nu \, V^* \, \sfJ  \, U \big)^2 \big\} 
\ \geq \ 
-\Tr\big\{ \big( k_ \nu \, |V| \, \sqrt{1+|V|^2} \big)^2 \big\} \, .
\end{align} 
Furthermore, if $U, V \geq 0$ and $\sfJ V\sfJ =V$ then 
\eqref{eq-III.01} is an equality.
%

\begin{proof}
For $\eps>0$ let
\begin{align} \label{eq-III.01,1}   
D_\eps := (\eps + VV^*) \, (1+VV^*)^{-1} 
\ \ \text{and} \ \ 
B_\eps := (\eps + V^*V) \, (1+V^*V) \, .
\end{align} 
Due to the spectral theorem and using $|V|^2 = V^*V$, 
\begin{align} \label{eq-III.01,2}   
D_\eps^\beta = (\eps + VV^*)^\beta (1+VV^*)^{-\beta} 
\ \ \text{and} \ \ 
B_\eps^\beta = (\eps + |V|^2)^\beta (1+|V|^2)^{\beta} \, ,
\end{align} 
for any $\beta \in [-1,1]\setminus\{0\}$. 
By \eqref{eq-II.08}-\eqref{eq-II.09}, we
have that
\begin{align} \label{eq-III.02}   
V V^* \, \sfJ U \ = \ V U^* \sfJ V \ = \ \sfJ U \, V^*V 
\quad \text{and} \quad
VV^* \, V \ = \ V \, V^*V \, ,
\end{align} 
which implies first that
\begin{align} \label{eq-III.03,1}   
(r + V V^*)^{-1} \, V \ = \ & V \, (r + V^*V)^{-1} \, ,
\\[1ex] \label{eq-III.03,2}   
(r + V V^*)^{-1} \, \sfJ U \ = \ & \sfJ U \, (r + V^* V)^{-1} \, ,
\end{align} 
for any $r >0$. Using the identity 
$A^{-\alpha} = \int_0^\infty (r + A)^{-1} 
\frac{\sin(\pi \alpha) \, dr}{\pi \, r^\alpha}$
for $\alpha \in (0,1)$ and strictly positive self-adjoint 
operators $A \geq \mu \bfone >0$ 
(see, e.g., \cite[Theorem~1.4.7]{Caps2002}) and 
Identity~\eqref{eq-III.03,1}, we then obtain the quadratic form estimate
\begin{align} \label{eq-III.04,1}   
V^* \, D_\eps^{-1/2} \, V 
\ = \ &
V^* \, (\eps + VV^*)^{-1/2} (1+VV^*)^{1/2} \, V 
\nonumber \\[1ex] 
\ = \ &
V^*V \, (\eps + V^*V)^{-1/2} (1+V^*V)^{1/2} 
\ \leq \ B_0^{1/2} \, .
\end{align} 
Similar to \eqref{eq-III.04,1}, we obtain from
\eqref{eq-III.03,2} and \eqref{eq-II.08} the identity
\begin{align} \label{eq-III.04,2}   
U^* \sfJ D_\eps^{1/2} \sfJ U 
\ = \ &
U^* \sfJ (\eps + VV^*)^{1/2} (1+VV^*)^{-1/2} \sfJ U  
\nonumber \\[1ex] 
\ = \ &
U^*U \, (\eps + V^*V)^{1/2} (1+V^*V)^{-1/2} 
\ = \ B_\eps^{1/2} \, .
\end{align} 
Next, let $M_U$ be a partial isometry in a polar decomposition $U =
M_U |U|$. Using again \eqref{eq-III.03,2} and \eqref{eq-II.08}, we
arrive at
\begin{align} \label{eq-III.04,3}   
D_\eps^{1/4} \sfJ U 
\ = \ &
(\eps + VV^*)^{1/4} (1+VV^*)^{-1/4} \sfJ U  
\\[1ex] \nonumber 
\ = \ &
\sfJ M_U |U| \, (\eps + V^*V)^{1/4} (1+V^*V)^{-1/4} 
\ = \ 
\sfJ M_U B_\eps^{1/4} \, .
\end{align} 
Since $VV^*$ is a self-adjoint trace-class operator, there exists an 
orthonormal basis $\{\phi_n\}_{n=1}^\infty \subseteq \fh$ of eigenvectors 
of $VV^*$, and, for every $N \in \bbN$, the rank-$N$ projection 
\begin{align} \label{eq-III.05}   
P_N \ := \ \sum_{n=1}^N |\phi_n\ra \la \phi_n|
\end{align} 
commutes with $D_\eps^{\beta}$. The Cauchy--Schwarz inequality
yields
\begin{align} \label{eq-III.06}   
\big| \Tr & [ P_N \sfJ U k_\nu U^* \sfJ V k_\nu V^*] \big|^2
\nonumber \\[1ex] 
\ = \ &
\big| \Tr[ P_N \, D_\eps^{1/4} \sfJ U k_\nu 
U^* \sfJ  D_\eps^{1/4} \, D_\eps^{-1/4} V k_\nu V^* D_\eps^{-1/4}] \big|^2
\ \leq \ 
R_1 \cdot R_2 \, ,
\end{align} 
where
\begin{align} \label{eq-III.07}   
R_2 \ := \ &
\Tr\Big[ D_\eps^{-1/4} V k_\nu V^* D_\eps^{-1/2} V k_\nu V^* D_\eps^{-1/4} ] 
\ \leq \ 
\Tr\Big[ D_\eps^{-1/4} V k_\nu B_0^{1/2} k_\nu V^* D_\eps^{-1/4} ] 
\nonumber \\[1ex]
\ = \ &
\Tr\Big[  B_0^{1/4} k_\nu V^* D_\eps^{-1/2} V k_\nu B_0^{1/4} ]
\ \leq \ 
\Tr\Big[ k_\nu B_0^{1/2} k_\nu B_0^{1/2} ] \, ,
\end{align} 
additionally using \eqref{eq-III.04,1} twice, and   
\begin{align} \label{eq-III.08}   
R_1 \ := \ &
\Tr\Big[ P_N \, D_\eps^{1/4} \sfJ U k_\nu U^* \sfJ D_\eps^{1/2} 
\sfJ U k_\nu U^* \sfJ \, D_\eps^{1/4}  P_N ] 
\nonumber \\[1ex]
\ = \ &
\Tr\Big[ P_N \, \sfJ M_U B_\eps^{1/4} k_\nu B_\eps^{1/2} k_\nu 
B_\eps^{1/4} M_U^* \sfJ \, P_N ] \, ,
\end{align} 
using \eqref{eq-III.04,2} and \eqref{eq-III.04,3}. Next, by the spectral
theorem and the fundamental theorem of calculus, we have that
\begin{align} \label{eq-III.09}   
0 \ \leq \ & B_\eps^\beta - B_0^\beta 
\ = \ 
(1+|V|^2)^{\beta/2} 
\Big( \int_0^\eps \frac{\beta \: d\tau}{(\tau+|V|^2)^{1-\beta}} \Big) 
(1+|V|^2)^{\beta/2}
\nonumber \\[1ex]
\ \leq \ &
 \eps^\beta (1+|V|^2)^\beta \, ,
\end{align} 
and hence 
\begin{align} \label{eq-III.10}   
\big\| B_\eps^\beta - B_0^\beta \big\|_\op 
\ \leq \  
 \eps^\beta (1+\|V\|^2)^\beta \, ,
\end{align} 
for any $\beta \in (0,1)$. It follows that 
\begin{align} \label{eq-III.11}   
\lim_{\eps \to 0} \big\{
\sfJ M_U B_\eps^{1/4} k_\nu B_\eps^{1/2} k_\nu B_\eps^{1/4} M_U^* \sfJ \big\}
\ = \ 
\sfJ M_U B_0^{1/4} k_\nu B_0^{1/2} k_\nu B_0^{1/4} M_U^* \sfJ 
\end{align} 
converges in operator norm. Since $P_N$ is of finite rank, the trace
on the right side of \eqref{eq-III.08} convergences, as well, namely
\begin{align} \label{eq-III.12}   
\lim_{\eps \to 0} 
\Tr[ & P_N \, \sfJ M_U B_\eps^{1/4} k_\nu B_\eps^{1/2} k_\nu 
B_\eps^{1/4} M_U^* \sfJ \, P_N ] 
\nonumber \\[1ex]
\ = \ &
\Tr[ P_N \, \sfJ M_U B_0^{1/4} k_\nu B_0^{1/2} k_\nu 
B_0^{1/4} M_U^* \sfJ \, P_N ] 
\\[1ex] \nonumber 
\ = \ &
\Tr[ B_0^{1/4} k_\nu B_0^{1/4} M_U^* \sfJ \, P_N \, 
\sfJ M_U B_0^{1/4} k_\nu B_0^{1/4} ] 
\ \leq \ 
\Tr[ B_0^{1/2} k_\nu B_0^{1/2} k_\nu ] \, ,
\end{align} 
using that 
$M_U^* \sfJ P_N \sfJ M_U \leq M_U^* \sfJ^2 M_U = M_U^* M_U \leq \bbone$. 
Inserting these estimates for $R_1$ and $R_2$ into \eqref{eq-III.06},
we conclude 
\begin{align} \label{eq-III.13}   
\big| \Tr[P_N \sfJ U k_\nu U^* \sfJ V k_\nu V^*] \big|
\ \leq \ 
\Tr[k_\nu B_0^{1/2} k_\nu B_0^{1/2}],
\end{align} 
for any $N \in \bbN$. Moreover, as a product of the bounded
operators $\sfJ U k_\nu U^* \sfJ$ and $k_\nu$ and the two 
Hilbert--Schmidt operators $V$ and $V^*$, the
operator $\sfJ U k_\nu U^* \sfJ V k_\nu V^*$ is
trace-class, and 
\begin{align} \label{eq-III.14}   
\big| \Tr[\sfJ U k_\nu U^* \sfJ V k_\nu V^*] \big| 
\ = \ &
\lim_{N \to \infty}
\big| \Tr[P_N \sfJ U k_\nu U^* \sfJ V k_\nu V^*] \big| 
\\[1ex] \nonumber
\ \leq \ & 
\Tr[k_\nu B_0^{1/2} k_\nu B_0^{1/2}] \, ,
\end{align} 
which yields the desired inequality. The second statement follows
directly from $\sfJ k_\nu \sfJ =-k_\nu$.
\end{proof}
\end{lem}
%

The other term in the BHF energy functional that contains $V^* \sfJ U$
is of the form 
$\la G_\nu + k_\nu \eta | V^* \sfJ U (G_\nu+k_\nu \eta) \ra$. It
obeys an inequality analogous to \eqref{eq-III.01} whose proof,
however, is simpler than the one for Lemma~\ref{lem-III.1} 
due to the fact that $|G_\nu+k_\nu \eta \ra \la G_\nu+k_\nu \eta |$ 
is of rank one and no infinite sums are involved. 
We omit the proof and only state the estimate in the following lemma.
%
\begin{lem}  \label{lem-III.2}
Let $B(U,V) \in \Bog[\fh]$ and $\eta \in \fh$. Then
\begin{align} \label{eq-III.15}    
\left|\big\la G_\nu +k_\nu \eta \big|\: V^* \sfJ U (G_\nu +k_\nu \eta)
\big\ra\right| 
\ \geq \ 
-\big\la G_\nu +k_\nu \eta \big| \: |V| \sqrt{1+|V|^2} 
        (G_\nu +k_\nu \eta) \big\ra \, .
\end{align} 
Furthermore, if $U, V\geq 0$, $\sfJ V\sfJ =V$ and $\sfJ \eta=\eta$ then
\eqref{eq-III.15} is an equality.
\end{lem}

\paragraph{Proof of Theorem~\ref{thm-I.1}.}
Recall from \eqref{eq-II.15} the energy functional
\begin{align} \label{eq-III.16}  
\tcE_{g,\vp}(U,V,\eta)
\ = \ 
\cS_{g,\vp}(V,\eta) + \cT(U,V) + \cQ_{g}(U,V,\eta) + \cI(V,\eta) \, ,
\end{align} 
and note that $U$ enters only $\cT$ and $\cQ_{g}$. Using
Lemmata~\ref{lem-III.1} and \ref{lem-III.2}, we obtain the
following lower bounds on the trace term
\begin{align} \label{eq-III.17}   
\cT(U,V) \ \geq \ 
\frac{1}{4} \sum_{\nu=1}^3 \Big\{ 
- \Tr\big[ (k_\nu |V| \sqrt{1+|V|^2} )^2 \big] +
\Tr\big[ k_\nu |V|^2 k_\nu (1+|V|^2) \big] \Big\} 
\end{align} 
and on the quadratic form term
\begin{align} \label{eq-III.18}   
\cQ_{g}(U,V,\eta) 
\ \geq \ &
\frac{1}{2} \sum_{\nu=1}^3 \Big\{ 
\big\la G_\nu+k_\nu \eta \: \big| \: (1+2|V|^2) (G_\nu+k_\nu \eta) \big\ra
\nonumber \\ &
-2 \big\la G_\nu + k_\nu \eta \: \big| \: |V|\sqrt{1+|V|^2} 
(G_\nu+k_\nu \eta) \big\ra \Big\} 
\\[1ex] \nonumber 
\ = \ & 
\frac{1}{2} \sum_{\nu=1}^3 \Big\{ 
\Big\la G_\nu + k_\nu \eta \: \Big| \: 
\big( \sqrt{1+|V|^2}-|V| \big)^2 (G_\nu+k_\nu \eta) \Big\ra \Big\} \, .
\end{align} 
This shows that, for every admissible $B(U,V) \in \Bog[\fh]$ and $\eta
\in \fh$, we can find a positive Hilbert--Schmidt operator $W := |V|$
such that $\cE_{g,\vp}(W,\eta) \leq \tcE_{g,\vp}(U,V,\eta)$, which
implies the lower bound $\LBOUND(\vp) \leq E_\BHF(\vp)$. The upper bound
$E_\BHF(\vp) \leq \UBOUND(\vp)$ follows immediately from the conditions 
for equality to hold in \eqref{eq-III.01} of Lemma~\ref{lem-III.1} 
and in \eqref{eq-III.15} of Lemma~\ref{lem-III.2}. 
\hspace*{\fill} $\square$

Note that the choice of $\sfJ$ and its special properties 
$\sfJ k_\nu \sfJ = - k_\nu$ and $\sfJ G_\nu=-G_\nu$ plays an important
role in establishing the upper bound, while the lower bound does not
require them. The resulting additional constraints $\sfJ V\sfJ =V$ and
$\eta=\sfJ \eta$ on the admissible variations in the upper bound lead to
\begin{align} \label{eq-III.19}   
\rRe\la \eta | G_\nu \ra
\ = \ 
\rRe \la \sfJ  G_\nu | \sfJ \eta\ra 
\ = \ 
-\rRe\la G_\nu|\eta\ra,
\end{align} 
implying that $\rRe\la \eta | G_\nu \ra = 0$. Similarly, under these
contraints we have that 
$\Tr[k_\nu V^2] = 0$ and $\la\eta |k_\nu \eta \ra = 0$. 
We conclude:
%
\begin{lem}\label{lem-III.3}
Let $V=\sfJ V \sfJ\in\cL^2_{\geq 0}(\fh)$ and $\eta =\sfJ \eta\in\fh$. 
Then the square term
\begin{align} \label{eq-III.20}  
\cS_{g,\vp}(V,\eta) \ = \ \vp^{\; 2} \, , 
\end{align} 
does not depend on $V$ and $\eta$, and in particular
\begin{align} \label{eq-III.21}   
\UBOUND(\vp) \ = \ \UBOUND(\vO) + \vp^{\; 2} \, .
\end{align} 
\end{lem}
%
Equation \eqref{eq-III.21} is consistent with the fact that $E_{\gs}(\vec{0})\leq E_\gs(\vp)$, for every $\vp\in\mathbb{R}^3$, as has been shown in \cite{Hiroshima2007}.
On the other hand, suppose that $\delta>0$ and $\vp\neq 0$ such that $B(\vp,\delta)\subseteq S_{\sigma,\Lambda} $, compare \eqref{eq-I.01}, and let
\begin{align}\label{eq-III.22}
\varphi_\delta:=\frac{\delta^{-3/2}}{4\pi}\bfone_{B(\vp,\delta)},
\end{align}
be a $L^2$-normalized localisation. By computing
\begin{align}\label{eq-III.22.1}
&\la a^*(\varphi_\delta)\Omega\:|\:H_{g,\vp} \:a^*(\varphi_\delta)\Omega \ra 
\\ \label{eq-III.22.2}
=& \
\frac{1}{2} \sum_{\nu=1}^3 \big\{\|( k_\nu - p_\nu ) \varphi_\delta \|^2 + 2\|G_\nu\otimes_{s} \varphi_\delta \|^2 +|\la  \varphi_\delta |G_\nu \ra|^2\big\}
\: + \: \big\la \varphi_\delta\big| |k| \varphi_\delta\big\ra
\\ \label{eq-III.22.3}
\leq &
\frac{1}{2} \delta^2+\frac{1}{2} \|\vG\|^2+\delta^3\frac{16\pi^2g^2}{|\vp|-\delta}+|\vp| + \delta
\end{align}
we obtain the well known upper bound
\begin{align}\label{eq-III.22.4}
E_\gs(\vp)\leq \frac{1}{2} \|\vG\|^2+|\vp|\,.
\end{align}
which shows that $\UBOUND(\vp)$ is an accurate upper bound on $E_\BHF(\vp)$ only for small values of $|\vp|<1$, if at all.
%

For this reason we henceforth focus on $\vp = \vO$ and analyze 
$\UBOUND(\vO)$. By \eqref{eq-III.20}, $\cS_{g,\vO}(V,\eta) = 0$ in this case,
and the terms $\cQ_{g}(\sqrt{1+|V|^2}, V, \eta)$ and $\cI(V,\eta)$ 
are quadratic in $\eta$ which, for a given $V = \sfJ V \sfJ \geq 0$,
allows us to determine the optimal choice 
$-\hA_V^{-1} \hxi_V = - \sfJ \hA_V^{-1} \hxi_V$ of $\eta$ explicitly by
completing a square. Indeed, according to \eqref{eq-III.18} and
\eqref{eq-II.19}, we have that

\begin{align} \label{eq-III.23}  
\cQ_{g}( & \sqrt{1+|V|^2}, V, \eta) + \cI(V,\eta) 
- \cQ_{g}(\sqrt{1+|V|^2}, V, 0) - \cI(V,0) 
\\[1ex] \nonumber 
& \ = \
\la \eta \, | \, \hA_V \eta \ra \, + \, 
2 \rRe \la \eta \, | \, \hxi_V \ra 
\ = \ 
\big\| \hA_V^{1/2} \eta + \hA_V^{-1/2} \hxi_V \big\|^2 \, - \,  
\big\la \hxi_V \, \big| \, \hA_V^{-1} \hxi_V \big\ra \, ,
\end{align} 
where 
\begin{align} \label{eq-III.24} 
\hA_V \ := \ & 
|k| + \frac{1}{2} \sum_{\nu=1}^3 
k_\nu \big( \sqrt{1+V^2} - V \big)^2 k_\nu \, ,
\\[1ex] \label{eq-III.25} 
\hxi_V \ := \ &
\frac{1}{2} \sum_{\nu=1}^3 
k_\nu \big( \sqrt{1+V^2} - V \big)^2 G_\nu \, .
\end{align} 
We introduce some new notation,
\begin{align} \label{eq-III.26}  
\hcE(V) \ := \ & \cE_{g,\vO}\big( V, - \hA_V^{-1} \hxi_V \big)
\\[1ex] \nonumber 
\ = \ & 
\frac{1}{2} \sum\nolimits_{\nu=1}^3 \Big\{\Tr[k_\nu V^2 k_\nu (1+V^2)]
-\Tr[(k_\nu V \sqrt{1+V^2})^2]
\\ \nonumber & 
\quad + \big\la G_\nu|(\sqrt{1+V^2}-V)^2 G_\nu\big\ra \Big\} 
+ \Tr\big[ \:|k| V^2 \big] - \big\la \hxi_V | \hA_V^{-1} \hxi_V \big\ra \, ,
\end{align} 
for the energy functional with the optimal choice of $\eta$ and
observe that 
%
\begin{lem}\label{lem-III.4}
At $\vp =\vO$, the upper bound in Theorem~\ref{thm-I.1} fulfills the
simplified variational characterisation
\begin{align} \label{eq-III.27}  
\UBOUND(\vO) \ = \ 
\inf\Big\{ \hcE(V)  \; \Big| 
\ V = \sfJ V \sfJ \in \cL_{\geq 0}^2(\fh) \Big\} \, .
\end{align} 
\end{lem}

\section{New Parameterization} \label{sec-IV}
We introduce the parameterization
\begin{align}\label{eq-IV.01}   
V \ = \ \sinh\Big[\tfrac{1}{2}\log(1+z)\Big]
\ = \ 
\frac{1}{2}\Big[(1+z)^\frac{1}{2}-(1+z)^{-\frac{1}{2}}\Big]
\ = \ 
\frac{z}{2\sqrt{1+z}} 
\end{align}
of $V$ in terms of a new operator $z$. In Lemma~\ref{lem-IV.1}~(i)
below we show that $V$ is a positive Hilbert--Schmidt operator obeying
$V = JVJ$ if, and only if, $z$ possesses these properties. The main
advantage of the variable $z$ over $V$ is that those functions of $V$
occuring in \eqref{eq-III.26}, i.e.,
\begin{align}\label{eq-IV.02}
V^2 \ = \ & \frac{z^2}{4(1+z)} 
\ = \ 
\frac{1}{4} \big[ (1+z) + (1+z)^{-1} -2 \big] \, ,
\\[1ex] \label{eq-IV.03}   
V\sqrt{1+V^2} \ = \ & \frac{z^2+2z}{4(1+z)}  
\ = \ 
\frac{1}{4} \big[ (1+z) - (1+z)^{-1} \big] \, ,
\\[1ex] \label{eq-IV.04}   
\big( \sqrt{1+V^2} - V \big)^2 \ = \ & \frac{1}{1+z} \, .
\end{align}
are sums of $1+z$ and $(1+z)^{-1}$ whose derivatives are 
easy to calculate explicitly. Indeed, the functional \eqref{eq-III.26}, 
expressed in terms of $z$, assumes the form
\begin{align}\label{eq-IV.05}
\cE(z) \ := \ &
\frac{1}{2} \sum\nolimits_{\nu=1}^3 
\Big\{\big\la G_\nu \big| \, (1+z)^{-1} G_\nu\big\ra 
- \frac{1}{4}\Tr\big[k_\nu z k_\nu z(1+z)^{-1}\big]
\Big\} 
\nonumber \\
& \quad + \frac{1}{4} \Tr\big[\:(|k|+\tfrac{1}{2}|k|^2)z^2(1+z)^{-1}\big] - \la {\xi}_z | {A}_z^{-1}{\xi}_z\ra,
\end{align}
where 
\begin{align}\label{eq-IV.06}
{A}_z \ = \ |k|+\frac{1}{2} \sum_{\nu=1}^3k_\nu (1+z)^{-1} k_\nu \, ,
\quad {\xi}_z \ = \ \frac{1}{2}\sum_{\nu=1}^3k_\nu(1+z)^{-1}G_\nu \, .
\end{align}
The following two lemmata show that the minimization of
$\hcE$ over $V$ yields the same value as the minimization of $\cE$
over $z$, thus justifying the passage from the variable $V$ to the
new variable $z$ when analyzing the upper bound
$\UBOUND(\vec{0})$.

For their formulation we define the 
real subspace $\cB_{s.a.}(\fh) \subseteq \cB(\fh)$ of self-adjoint 
bounded operators and 
its convex subset $\cB_{>-1}(\fh) \subseteq \cB_{s.a.}(\fh)$
of bounded self-adjoint operators with smallest spectral value strictly
bigger than $-1$.
\begin{align}\label{eq-IV.06,1}
\cB_{s.a.}(\fh) \ = \ &
\big\{ z \in \cB(\fh) \: \big| \ z = z^* \big\} \, ,
\\[1ex] \label{eq-IV.06,2}
\cB_{>-1}(\fh) \ = \ &
\big\{ z \in \cB_{s.a.}(\fh) \: \big| \ \exists \, \mu > -1: \ \ 
z \geq \mu \, \bfone \big\} \, .
\end{align}
%
\begin{lem} \label{lem-IV.1}
Define the map $V_{(\cdot)} : \cB_{>-1}(\fh) \to \cB_{s.a.}(\fh)$ by
\begin{align}\label{eq-IV.06,3}
V_z \ := \ \sinh\Big[\tfrac{1}{2} \log(1+z) \Big] \, .
\end{align}
Then the following statements hold true. 
\begin{enumerate}[label=(\roman*)]
\item If $z \in \cL_{>-1}^2(\fh)$ then $V_z \in \cL_{s.a.}^2(\fh)$, 
and the restriction 
$(z \mapsto V_z): \cL_{>-1}^2(\fh) \to \cL_{s.a.}^2(\fh)$ 
is a bijection.

\item If $z \in \cL_{\geq 0}^2(\fh)$ then $V_z \in \cL_{\geq 0}^2(\fh)$, 
and the restriction 
$(z \mapsto V_z): \cL_{\geq 0}^2(\fh) \to \cL_{\geq 0}^2(\fh)$ 
is a bijection.
\end{enumerate}
\begin{proof} 
Assertion~\textit{(ii)} is Lemma~IV.10~(i) of \cite{BachHach2022}
  with $y:=z+1$, and also \textit{(i)} is proven similarly.
Specifically, to establish Assertion~\textit{(i)} we observe that
\begin{align}\label{eq-IV.06,4}
f: \ (-1,\infty) \ \to \ (-\infty,\infty) \, ,
\quad \lambda \mapsto \sinh\Big[ \tfrac{1}{2} \log(1+\lambda) \Big]
\end{align}
%
%
is a strict monotonically increasing homeomorphism. By the
Hilbert--Schmidt theorem, the lowest eigenvalue
$\lambda_\rmmin$ of $z \in \cL_{>-1}^2(\fh)$ obeys
$\lambda_\rmmin >-1$, and the identity 
$4 V^2 = z^2 (1+z)^{-1}$ implies that
\begin{align}\label{eq-IV.07}
4 \, \Tr [V^2] 
\ \leq \ 
(1+\lambda_\rmmin)^{-1} \Tr[z^2]
\ < \ \infty \, .
\end{align}
Conversely, if $V \in \cL_{s.a.}^2(\fh)$, then 
\begin{align}\label{eq-IV.08}
V \ = \ \sum_{n=1}^\infty V_n \, |\psi_n \ra \la \psi_n| \, ,
\end{align}
where $(\psi_n)_{n=1}^\infty \subseteq \fh$ is an orthonormal basis
of eigenvectors of $V$ and its eigenvalues $V_n$ converge $V_n \to 0$, as $n\to \infty$. By the spectral theorem 
$V=f(z)$, with
\begin{align}\label{eq-IV.09}
z \ = \ \sum_{n} f^{-1}(V_n) |\psi_n \rangle \langle \psi_n |,
\end{align}
so that the lowest eigenvalue $\lambda_\rmmin$ of $z$
obeys $\lambda_\rmmin = f^{-1}(V_\rmmin) > -1$, where 
$V_\rmmin$ is the lowest eigenvalue of $V$. Therefore
\begin{align}\label{eq-IV.10}
0 \ < \ (1+z)^{-1/2} \ \leq \ (1+\lambda_\rmmin)^{-1/2} \, ,
\end{align}
and thus
\begin{align}\label{eq-IV.11}
1 \ \leq \ (1+z)^{1/2} 
\ = \ 
(1+z)^{-1/2} +2V
\ \leq \ 
(1+\lambda_\rmmin)^{-1/2} +2 \|V\|_\op \, ,
\end{align}
which implies that
\begin{align}\label{eq-IV.12}
\Tr[z^2] \ \leq \ 
4 \big\{ (1+\lambda_\rmmin)^{-1/2} + 2 \|V\|_\op \big\}^2 \, \Tr[V^2] 
\ < \ \infty \, .
\end{align}
That the map $z \mapsto V_z$ is injective and surjective follows
immediately from the spectral theorem and the respective properties of
the homeomorphim $f$.
\end{proof}
\end{lem}

\begin{lem} \label{lem-IV.2}
The following variational problems are equivalent
\begin{align} \label{eq-IV.13}
E_\mathrm{up}(\vec{0}) \ = \ X \ = \ Y \ = \ Z \, ,
\end{align}
where
\begin{align} \label{eq-IV.14} 
X \ := \ &
\inf\big\{ \cE(z) \: \big| 
\: z = \sfJ z\sfJ \in \cL_{\geq 0}^2(\fh)\big\} \,, 
\\[1ex] \label{eq-IV.15}
Y \ := \ &
\inf\big\{ \hcE(V) \: \big| 
\: V = \sfJ V \sfJ \in \cL_{s.a.}^2(\fh) \big\} \, ,
\\[1ex] \label{eq-IV.16}
Z \ = \ &
\inf\big\{ \cE(z) \: \big| 
\: z = \sfJ z \sfJ \in \cL_{>-1}^2(\fh) \big\} \, .
\end{align}
\begin{proof}
The properties $\sfJ V\sfJ=V$ and $\sfJ z\sfJ=z$ imply one another by virtue of the spectral theorem, hence
Lemma~\ref{lem-IV.1}~\textit{(i)} and \textit{(ii)} immediately 
yield
\begin{align} \label{eq-IV.17}
E_\mathrm{up}(\vec{0}) \ = \ X 
\quad \text{and} \quad Y \ = \ Z \, .
\end{align}
Clearly, $Y \leq E_\mathrm{up}(\vec{0})$ is true, and it remains to
show that $Y \geq E_\mathrm{up}(\vec{0})$. This is not immediate
from \eqref{eq-III.26}, but going back to \eqref{eq-I.23}, i.e.,
before minimizing w.r.t.\ $\eta$, the only terms that depend on the sign of $V$ are
\begin{align} \label{eq-IV.18}
- \Tr\big[ (k_\nu V\sqrt{1+V^2} )^2 \big] 
\ \ \text{and} \ \ 
\big\la G_\nu +k_\nu \eta \big| \big(\sqrt{1+V^2}-V\big)^2 
(G_\nu +k_\nu \eta) \big\ra \, . 
\end{align}
To analyze the second term in \eqref{eq-IV.18},
we observe that the function $h(x):=(\sqrt{1+x^2}-x)^2$
is monotonically decreasing on $\bbR$, since its derivative
\begin{align}\label{eq-IV.19}
\frac{1}{2} h'(x) \ = \ -\frac{(\sqrt{1+x^2}-x)^2}{\sqrt{1+x^2}} \ < \ 0 
\end{align}
is manifestly negative. Therefore, 
\begin{align}\label{eq-IV.20}
\big\la G_\nu +k_\nu \eta \big| \: h(V) (G_\nu +k_\nu \eta) \big\ra 
\ \geq \ 
\big\la G_\nu +k_\nu \eta \big| \: h(|V|) (G_\nu +k_\nu \eta) \big\ra \, .
\end{align}
On the other hand, setting $U:=\sqrt{1+V^2}$, we obtain
\begin{align}\label{eq-IV.21}
- \Tr\big[(k_\nu V\sqrt{1+V^2} )^2\big]= \Tr\big[(k_\nu V^\star JU )^2\big] \geq -\Tr \big[(k_\nu |V|\sqrt{1+V^2} )^2\big],
\end{align}
from Lemma~\ref{lem-III.1}. Hence, passing from $V \in
\cL_{{s.a.}}^2(\fh)$ to $|V| \in \cL_{{s.a.}}^2(\fh)$ always decreases
the functional, and we can now optimize in $\eta$ to get the
desired result.
\end{proof}
\end{lem}
%
The extension from $z\geq 0$ to $z>-1$ is useful because then $z=0$
becomes an inner point of the set of admissible operators, and the
derivative is the usual Fréchet derivative. Note that replacing $Z$ in
Lemma~\ref{lem-IV.2} by
\begin{align}\label{eq-IV.22}
Z_\kappa \ = \ &
\inf\big\{ \cE(z) \: \big| 
\: z = \sfJ z \sfJ \in \cL_{>\kappa}^2(\fh) \big\} \, ,
\end{align}
for any $-1<\kappa<0$, still yields $\UBOUND=Z_\kappa$, with almost
the same proof. Throughout Sect.~\ref{sec-V}, we assume that
$z\in\cL_{> -1/2}^2(\fh)$, so that $\|(1+z)^{-1}\|_\op$ is universally
bounded by 2.

\section{Stationarity Condition for the Upper Bound} \label{sec-V}
%
This section is devoted to the proof of Theorem~\ref{thm-I.2}, i.e., 
the derivation of the Fréchet derivative of $z \mapsto \cE(z)$
on $\cL_{> -1/2}^2(\fh)$, which is an open convex subset of the
real vector space of self-adjoint Hilbert--Schmidt operators 
on $\fh$. Recall that the derivative of $\cE(z)$ is the Hilbert--Schmidt 
operator $\partial_z \cE(z)\in \cL_{s.a.}^2(\fh)$ such that
\begin{align} \label{eq-V.01}   
\cE(z+h) - \cE(z)
\ = \ 
\Tr\big\{ \partial_z\cE(z) \, h \big\} 
\, + \, o\big(\|h\|_{\cL^2(\fh)}\big) \, ,
\end{align} 
provided $z+h\in\cL_{> -1/2}^2(\fh)$.

\paragraph{Proof of Theorem~\ref{thm-I.2}.}
We denote $M := \|z\|_{\cL^2(\fh)} + 1$ and frequently make use of the
estimates $\| k_\nu \|_\op \leq \Lambda$ and
\begin{align} \label{eq-V.02}   
\bigg\| \frac{1}{1+z+h} \, - \, \frac{1}{1+z} \, + \, 
\frac{1}{1+z} \, h \, \frac{1}{1+z} \bigg\|_{\cL^2(\fh)} 
\ \leq \ 8 \| h \|_{\cL^2(\fh)}^2 \,  ,
\end{align} 
which follows from the second resolvent identity and 
$\big\| (1+z)^{-1} \big\|_\op \leq 2$, 
$\big\| (1+z+h)^{-1} \big\|_\op \leq 2$.

We compute the derivative of $\cE(z)$. First, we obtain from
\eqref{eq-V.02} that
\begin{align} \label{eq-V.04}   
\Tr\Big[k_\nu & (z+h) k_\nu \frac{z+h}{1+z+h} \Big]
\, - \, \Tr\Big[k_\nu z k_\nu \frac{z}{1+z}\Big] 
\\ \nonumber 
\ = \ &
\Tr\Big[k_\nu h k_\nu \frac{z}{1+z}\Big] 
\, + \,
\Tr\Big[k_\nu (z+h) k_\nu \Big( \frac{1}{1+z} - \frac{1}{1+z+h} \Big) \Big]
\\[1ex] \nonumber 
\ = \ &
\Tr\Big[k_\nu h k_\nu \frac{z}{1+z}\Big] 
\, + \,
\Tr\Big[k_\nu z k_\nu \frac{1}{1+z} h \frac{1}{1+z} \Big] 
\, + \, \cO\big( \| h \|_{\cL^2(\fh)}^2 \big) 
\\[1ex] \nonumber 
\ = \ &
\Tr\Big[ \Big( k_\nu \frac{z}{1+z} k_\nu 
+ \frac{1}{1+z} k_\nu z k_\nu \frac{1}{1+z} \Big) h \Big]
\, + \, \cO\big( \| h \|_{\cL^2(\fh)}^2 \big) \, .
\end{align} 
Recalling that $A_0 = |k| + \tfrac{1}{2}|k|^2$,
we further obtain
\begin{align} \label{eq-V.05}   
\Tr\big[ A_0 & (z+h)^2 (1+z+h)^{-1}] 
\, - \, \Tr\big[ A_0 z^2 (1+z)^{-1}] 
\\[1ex] \nonumber 
\ = \ &
\Tr\big[ A_0 \, (zh + hz) (1+z)^{-1}] 
\, - \, \Tr\Big[A_0 \,z^2 \frac{1}{1+z}h\frac{1}{1+z}  \Big]
\, + \, \cO\big( \| h \|_{\cL^2(\fh)}^2 \big) 
\\[1ex] \nonumber 
\ = \ &
\Tr\Big[ \Big( \frac{1}{1+z} A_0 \frac{z}{1+z} 
+ \frac{z}{1+z} A_0 \Big) h \Big] 
\, + \, \cO\big( \| h \|_{\cL^2(\fh)}^2 \big) \, .
\end{align} 
From $A_z = A_0 - \frac{1}{2} \sum_{\nu=1}^3 k_\nu \frac{z}{1+z} k_\nu$
and Eqs.~\eqref{eq-V.04} and \eqref{eq-V.05} follows that
\begin{align} \label{eq-V.05,1}   
\Tr\big[ A_0 & (z+h)^2 (1+z+h)^{-1}] 
\, - \, \Tr\big[ A_0 \, z^2 (1+z)^{-1}] 
\nonumber \\ 
& \, - \, \frac{1}{2}
\sum_{\nu=1}^3 \Tr\Big[ k_\nu (z+h) k_\nu \frac{z+h}{1+z+h}
\, - \, k_\nu z k_\nu \frac{z}{1+z} \Big]
\nonumber \\[1ex] 
\ = \ &
\Tr\Big[ \frac{1}{1+z} \Big( (1+z) A_z (1+z) - A_0 -\frac{1}{2}\sum_{\nu=1}^3 k_\nu z k_\nu \Big) \frac{1}{1+z} h \Big] 
\nonumber \\ 
& \quad
\, + \, \cO\big( \| h \|_{\cL^2(\fh)}^2 \big) \, .
\end{align} 
Next, we note that
\begin{align} \label{eq-V.06}   
\big\la \xi_{z+h} & \big| \, A_{z+h}^{-1} \xi_{z+h} \big\ra 
\, - \, \big\la \xi_{z} \big| \, A_{z}^{-1} \xi_{z} \big\ra 
\nonumber \\[1ex] 
\ = \ &
\big\la (\xi_{z+h} - \xi_{z}) \big| \, A_{z+h}^{-1} \xi_{z+h} \big\ra 
\, + \, \big\la \xi_{z} \big| \, (A_{z+h}^{-1} - A_{z}^{-1}) \xi_{z+h} \big\ra 
\nonumber \\ & \quad 
\, + \, \big\la \xi_{z} \big| \, A_{z}^{-1} (\xi_{z+h} - \xi_{z}) \big\ra \, .
\end{align} 
Using
\begin{align} \label{eq-V.07} 
A_{z+h}^{-1} - A_z^{-1} 
\ = \ 
\frac{1}{2} \sum_{\nu=1}^3 
A_z^{-1} k_\nu \frac{1}{1+z} h \frac{1}{1+z} k_\nu A_z^{-1}
\, + \, \cO\big( \| h \|_{\cL^2(\fh)}^2 \big) 
\end{align} 
and 
\begin{align} \label{eq-V.08} 
\xi_{z+h} - \xi_z
\ = \ 
- \frac{1}{2} \sum_{\nu=1}^3 
k_\nu \frac{1}{1+z} h \frac{1}{1+z} G_\nu 
\, + \, \cO\big( \| h \|_{\cL^2(\fh)}^2 \big) \, , 
\end{align} 
we obtain
\begin{align} \label{eq-V.09}   
\big\la (\xi_{z+h} & - \xi_{z}) \big| \, A_{z+h}^{-1} \xi_{z+h} \big\ra 
\, + \, \big\la \xi_{z} \big| \, A_{z}^{-1} (\xi_{z+h} - \xi_{z}) \big\ra 
\nonumber \\[1ex] 
\ = \ &
- \sum_{\nu=1}^3 \rRe \big\la k_\nu \frac{1}{1+z} h \frac{1}{1+z} G_\nu 
\big| \, A_{z}^{-1} \xi_{z} \big\ra 
\, + \, \cO\big( \| h \|_{\cL^2(\fh)}^2 \big)
\\[1ex] \nonumber 
\ = \ &
\Tr\Big[ \frac{1}{1+z} \rRe\Big( - \sum_{\nu=1}^3 
\big| G_\nu \big\ra \big\la k_\nu A_{z}^{-1} \xi_{z} \big| \Big)
\frac{1}{1+z} h \Big]  
\, + \, \cO\big( \| h \|_{\cL^2(\fh)}^2 \big) \, ,
\end{align} 
and
\begin{align} \label{eq-V.10}   
\big\la \xi_{z} & \big| \, (A_{z+h}^{-1} - A_{z}^{-1}) \xi_{z+h} \big\ra 
\\[1ex] \nonumber 
\ = \ &
 \frac{1}{2} \sum_{\nu=1}^3 
\big\la \xi_z \big| \, A_z^{-1} k_\nu \frac{1}{1+z} h 
\frac{1}{1+z} k_\nu A_z^{-1} \xi_{z} \big\ra 
\, + \, \cO\big( \| h \|_{\cL^2(\fh)}^2 \big) 
\\[1ex] \nonumber 
\ = \ &
\Tr\Big[ \frac{1}{1+z} \Big( \frac{1}{2} \sum_{\nu=1}^3 
\big| k_\nu A_{z}^{-1} \xi_{z} \big\ra \big\la k_\nu A_{z}^{-1} \xi_{z} \big| 
\Big) \frac{1}{1+z} h \Big]  
\, + \, \cO\big( \| h \|_{\cL^2(\fh)}^2 \big) \, .
\end{align} 
Furthermore, 
\begin{align} \label{eq-V.03}   
\big\la G_\nu \, \big| \:(1+z & +h)^{-1} G_\nu\big\ra
\, - \, \big\la G_\nu \, \big| \:(1+z)^{-1} G_\nu\big\ra
\\[1ex] \nonumber
\ = \ &
- \big\la (1+z)^{-1} G_\nu \, \big| \; h \, (1+z)^{-1} G_\nu\big\ra
\, + \, \cO\big( \| h \|_{\cL^2(\fh)}^2 \big) 
\\ \nonumber
\ = \ &
\Tr\Big[ \Big( - \frac{1}{1+z} | G_\nu \ra 
\la G_\nu |  \frac{1}{1+z} \Big) \, h \Big] 
\, + \, \cO\big( \| h \|_{\cL^2(\fh)}^2 \big) \, ,
\end{align} 
and adding up this and \eqref{eq-V.09} and \eqref{eq-V.10},
we obtain
\begin{align} \label{eq-V.11}   
& \frac{1}{2} \sum_{\nu=1}^3 \Big\{ 
\big\la G_\nu \, \big| \:(1+z+h)^{-1} G_\nu\big\ra
\, - \, \big\la G_\nu \, \big| \:(1+z)^{-1} G_\nu\big\ra \Big\}
\\[1ex] \nonumber & \qquad \qquad
\, - \, \big\la \xi_{z+h} \big| \, A_{z+h}^{-1} \xi_{z+h} \big\ra 
\, + \, \big\la \xi_{z} \big| \, A_{z}^{-1} \xi_{z} \big\ra 
\\[1ex] \nonumber
& \ = \ 
-\Tr\Big[ \frac{1}{1+z} \Big( 
\frac{1}{2} \sum_{\nu=1}^3 | G_\nu - k_\nu A_{z}^{-1} \xi_{z} \ra 
\la G_\nu - k_\nu A_{z}^{-1} \xi_{z} |  \Big) \frac{1}{1+z} h \Big] 
+ \cO\big( \| h \|_{\cL^2(\fh)}^2 \big) \, .
\end{align} 
Eqs.~\eqref{eq-V.11} and \eqref{eq-V.05,1} together yield
\begin{align} \label{eq-V.12}   
(1+z) & \partial_z\cE_g(z) (1+z) 
\ = \ 
-\frac{1}{2} \sum_{\nu=1}^3 \Big\{ | G_\nu - k_\nu A_{z}^{-1} \xi_{z} \ra 
\la G_\nu - k_\nu A_{z}^{-1} \xi_{z} | 
\\[1ex] \nonumber &
\, + \, 
\frac{1}{4} (1+z) A_z (1+z) - A_0 -\frac{1}{2}\sum_{\nu=1}^3 k_\nu z k_\nu 
\end{align} 
and, by identifying $\eta_z\equiv-A_z^{-1}\xi_z$ from \eqref{eq-I.24,2}, Theorem~\ref{thm-I.2}. 
\hspace{\fill} $\square$

\subsection*{Acknowledgement}
Support from DFG Grant Nr.~BA~1477/15-1, Project Nr.~505496137, of the
German Science Foundation is gratefully acknowledged. Furthermore, the
authors thank M.~Ballesteros, S.~Breteaux, and M.~Mlinarzik for useful
discussions.


\end{document}